\documentclass{article} 
\usepackage{hyperref}
\usepackage{auto-pst-pdf}

\usepackage{natbib}
 \bibpunct[, ]{(}{)}{,}{a}{}{,}%
\usepackage{algorithm,algorithmic}
\usepackage{subcaption}
\usepackage{enumerate}
\newcommand{\eqr}[1]{Eq.~\eqref{#1}}

\usepackage{xspace}
\usepackage{amsfonts}
\usepackage{url}
\usepackage{verbatim}
\usepackage{amsmath}
\usepackage{mathtools}  
\usepackage{amsmath}
\usepackage{amssymb}
\usepackage{algorithmic}
\usepackage{authblk}
\usepackage{algorithm,algorithmic}
\usepackage{amsthm}
\usepackage{subcaption}
\usepackage{enumerate}

\newtheorem{theorem}{Theorem}

\newtheorem{corollary}{Corollary}
\newtheorem{definition}{Definition}
\newtheorem{example}{Example}

\newtheorem{lemma}{Lemma}
\newtheorem{assumption}{Assumption}

\title{Equilibrium and Learning in Queues with Advance Reservations}
\author{Eran Simhon and David Starobinski \\  Boston University, College of Engineering}
\begin{document}
\maketitle





\begin{abstract}
Consider a multi-class preemptive-resume $M/D/1$ queueing system that supports advance reservations (AR). In this system, strategic customers must decide whether to reserve a server in advance (thereby gaining higher priority) or avoid AR. Reserving a server in advance bears a cost. In this paper, we conduct a game-theoretic analysis of this system, characterizing the equilibrium strategies. Specifically, we show that the game has two types of equilibria. In one type, none of the customers makes reservation. In the other type, only customers that realize early enough that they will need service make reservations. We show that the types and number of equilibria depend on the parameters of the queue and on the reservation cost. Specifically, we prove that the equilibrium is unique if the server utilization is below 1/2. Otherwise, there may be multiple equilibria depending on the reservation cost.  Next, we assume that the reservation cost is a fee set by the provider. In that case, we show that the revenue maximizing fee leads to a unique equilibrium if the utilization is below 2/3, but multiple equilibria if the utilization exceeds 2/3. Finally, we study a dynamic version of the game, where users learn and adapt their strategies based on observations of past actions or strategies of other users. Depending on the type of learning (i.e., action learning vs.\ strategy learning), we show that the game converges to an equilibrium in some cases, while it cycles in other cases.
\end{abstract}







\setlength{\pdfpagewidth}{8.5in}

\setlength{\pdfpageheight}{11in}

\section{Introduction} \label{sec:intro}
Many services, such as health care, cloud computing and banking, combine both a first-come-first-served policy and advance reservations (AR). \ Advance reservations benefit a service provider since knowledge about future demand can improve resource management and  quality-of-service (e.g.,~\cite{charbonneau2012survey}). Customers are also motivated to reserve in advance, since it decreases their expected waiting time. However, typically,  reservations bear an additional cost for customers.  This cost can be a reservation fee, the time or resources required for making the reservation,  the cost of financing advance payment, or the cost of cancellation if needed.

Since the decision of a customer, about reserving a server in advance or not, affects the waiting time of other customers, game theory is the solution of choice for studying such systems. Although there exists a rich literature on advance reservations, works that study advance reservation systems as a game are rare. The strategic behavior of customers in a system that support AR is studied in \cite{simhon2014game} and \cite{simhon2015pricing}. These two papers study a \textit{loss system}, i.e., a system with no queue. In this paper, instead, we focus on a queueing system (i.e., customers that encounter a busy server wait for service).
This leads to a different model and, interestingly, more explicit results. We show that the server utilization (traffic load) plays in key role in the behavior of the system and, specifically, in the number of equilibria.  
 
 We assume that the time axis is divided into two time-periods: a \textit{reservation period} and a \textit{service period}. This restriction simplifies the analysis and is common in the literature of advance reservations (e.g.,  \cite{ virtamo1992model}, \cite{yessad2007r} and \cite{syed2008t}). It can also be found in real life applications. For example, some service providers do not allow same-day-reservations.
 
During the reservation period, each customer realizes that he/she will need service at a specific future time point. Upon such a realization, the customer decides whether or not to make a reservation. Customers are assumed to be strategic and rational. Thus, a customer will make a reservation only if it reduces his/her expected total cost which consists of the reservation cost (if making a reservation) and the cost of waiting.

We start the analysis by finding the equilibrium structure of the game. We show that there are two possible types of equilibria. In the first type, none of the customers makes AR, while in the second type customers that realize early enough that they will need future service make AR. We refer to those two types of equilibria as \textit{none-make-AR} and \textit{some-make-AR}, respectively. We show that if the utilization of the queue (i.e., the ratio between the arrival rate and the service rate) is smaller than $1/2$, then the game has a unique equilibrium. Low AR costs lead to a \textit{some-make-AR} equilibrium, while high AR cost lead to a \textit{none-make-AR} equilibrium. If the utilization is greater than $1/2$, however, there also exists a middle range of AR cost such that any cost in that range leads to three equilibria, namely one \textit{none-make-AR} and two \textit{some-make-AR} equilibria.

Next, we assume that the AR cost is a fee charged by the service provider.  We analyze the game from the prospective of a provider aiming to maximize its revenue from AR fees. We show that if the utilization is greater than $2/3$, then the revenue maximizing fee leads to multiple equilibria. Thus, charging that fee may yield the highest possible revenue for the provider but possibly also no revenue. 

Finally, we study a dynamic version of the game.  We use \textit{best response dynamics} (as in \cite{fudenberg1998theory}) and distinguish between \textit{strategy-learning} and \textit{action-learning}. In \textit{strategy-learning}, customers obtain information about  strategies adopted at previous steps, while in \textit{action-learning}, customers \textit{estimate} the previous strategies by obtaining information about the actions taken at previous steps. Our analysis shows that starting with any initial belief about customers behavior \textit{(i)}~when implementing \textit{strategy-learning}, the system always converges to an equilibrium; \textit{(ii)}~when implementing \textit{action-learning}, the system converges to a \textit{none-make-AR} equilibrium if it exists and cycles otherwise;  \textit{(iii)}~if the equilibrium is unique, more customers, on average, make reservations under \textit{action-learning} than under  \textit{strategy-learning}.

The rest of the paper is structured as follows. In Section \ref{related_work}, we review related work. In Section \ref{sec:The model}, we formally define the game. In Section \ref{sec_analysis}, we find the equilibrium structure of the game. In Section \ref{sec:rev}, we derive the revenue maximizing fee and resulting equilibria. In Section \ref{sec_learning}, we define and analyze dynamic versions of the game. Section \ref{conclusions} concludes the paper and suggests directions for future research. 

\section{Related Work} \label{related_work}
Strategic behavior in queues (also known as queueing games) was pioneered by \cite{naor1969regulation} and has been studied extensively since. In that seminal paper, the author studies an $M/M/1$ queue where customers  decide whether to join or balk after observing the queue length. \cite{hassin2003queue} and \cite{hassin2016queue} conduct an extensive review of the field of queueing games. Most related to our work, \cite{balachandran1972purchasing} analyzes  strategic behavior in priority queues and \cite{qiu2016managing} and \cite{hayel2016decentralized} study strategic behavior in $M/D/1$ queues. None of these works consider advance reservations. 

Advance reservations have been researched from various other perspectives in the literature, including scheduling and routing algorithms for communication networks, methods for revenue maximization, and performance analysis of queueing systems. The work in~\cite{wang2013dynamic} describes a distributed architecture for advance reservation, while~\cite{smith2000scheduling} proposes a scheduling model that supports AR and evaluates several performance metrics.
The work in~\cite{virtamo1992model} analyzes the impact of advance reservations on server utilization under a stochastic arrival model, and~\cite{guerin2000networks} analyzes the effect of AR on the complexity of path selection. In \cite{weatherford1998tutorial}, the author reviews models for revenue management of perishable assets, such as airline seats and hotel rooms, that extend to various industries. The work in~\cite{reiman2008asymptotically} considers admission control strategies in reservation systems with different classes of customers, while~\cite{bertsimas2003restaurant} deals with policies for accepting or rejecting restaurant reservations.
The effects of overbooking, cancellations and regrets on advance reservations are studied in~\cite{liberman1978hotel,quan2002price,nasiry2012advance}. None of these prior works considers the \emph{strategic} behavior of customers in making AR, namely, that decisions of customers are not only influenced by prices and policies set by providers but also by their beliefs about decisions of other customers.

\cite{simhon2014game} introduces AR games. In that paper, the authors consider a loss system (i.e., customers that finds all servers busy leave). The authors show that the game may have multiple equilibria, where in  one equilibrium the number of reservations is a random variable, while in the other equilibrium, none of the customers makes reservation. In \cite{simhon2015pricing}, the authors study a dynamic version of the game. The main difference between the model of our paper and the model presented in \cite{simhon2014game} and \cite{simhon2015pricing}  is that our paper focus on a queuing system, while these papers focus on a loss system. Specifically, our paper shows that the server utilization plays a key role in determining the number and structure of equilibria. The characterization of the equilibrium strategies in our paper is also much more explicit than that provided in~\cite{simhon2014game}. 

The concept of learning an equilibrium is rooted in Cournot's duopoly model \cite{cournot1897recherches} and has been extensively researched since. Traditionally, learning models are used for fixed-player games (i.e., the same players participate at each iteration), see~\cite{lakshmivarahan1981learning,fudenberg1998theory} and~\cite{milgrom1991adaptive}. Several papers have focused on learning under stochastic settings. For example, in \cite{liu2011strategic} customers choose between buying a product at full price or waiting for a discount period. Decisions are made based on observing past capacities. \cite{altman1998individual} analyze a processor sharing model. In this model, customers choose between joining or balking after observing the history. \cite{zohar2002adaptive} present a model of abandonment from unobservable queues. The decision is based on the expected waiting time which is formed through accumulated experience. \cite{fu2009learning} assume that the same set of players participate in a bid for wireless resources at each stage. However, the number of packets that need to be transmitted at each iteration is a random variable.

Different learning models differ by their learning rules. A learning rule defines what kind of information players gain and how they use it. In this paper, we focus on  \textit{best response dynamics}. According to this rule, which is rooted in Cournot's work, players observe the most recent  actions adopted by other players and  assume that the same actions will be adopted at the next step. Another popular learning rule is \emph{fictitious play} which assumes that at each iteration, players observe actions made by other players at all previous steps and best-respond to the empirical frequency of observed actions. This rule was suggested by \cite{brown1951iterative}. In contrast, \cite{littman1994markov} and \cite{tan1993multi} assume that players only observe their own payoffs and learn by trial  and error.  \textit{Reinforcement learning} is an example of such a learning rule.

Other relevant work includes~\cite{niu2012pricing}, which presents a theoretical model for pricing cloud bandwidth reservations, in order to maximize social welfare. The reservation fee of each customer is a function of his/her guaranteed portion instead of the actual amount of resources reserved, as considered in our models as well as many practical services. In~\cite{menache14demand}, the authors consider the problem of deciding which type of service a customer should buy from a cloud provider. More specifically, that study considers two options: on-demand, which means paying a fixed price for service, and spot, a service offered by Amazon EC2 that allows users to bid for spare instances. They propose a no-regret online learning algorithm to find the best policy. Our approach complements this work in several ways. First, our framework considers advance reservations (similar to Reserved Instances in Amazon EC2). Second, our models integrate the strategic behavior of all participants (i.e., both the customers and the provider).

\section{Game Description}\label{sec:The model}
We consider a preemptive-resume $M/D/1$ queue that supports advance reservations. In our model, there is a reservation period which covers $[-T,0]$. Each customer $k=1,2,...$ is associated with a request time $-T\le t_k \le 0$ and a desired service starting time (shortly noted as arrival time) $s_k>0$. That is, if $t_1<t_2$, then customer $1$ has the opportunity to reserve the server before  customer $2$. If $s_1<s_2$, then customer $1$ wishes to be served earlier than customer $2$. The service period starts only after the reservation period ends. The request time can be interpreted as how much time in advance a customer realizes that he/she will need service at a future time point.

The request times are derived from a general continuous distribution  with cumulative distribution function $F_T(\cdot)$. The arrivals follow a Poisson process with rate $\lambda$. The service time is  $1/\mu$ and we assume that  $\lambda < \mu$.

Each customer, at his/her request time, decides whether to make a reservation or not. We denote those two actions by $AR$ and $AR'$, respectively. If a customer makes a reservation but his/her desired service time is already reserved, the nearest future available time will be reserved for that customer. A customer that does not make a reservation is served on a first-come-first-served basis along periods of times over which the server is not reserved.

The total cost of each customer consists of the reservation cost $C$ (if making AR) and the cost of waiting which is a linear function of the waiting time. Without loss of generality, we assume that the cost of waiting is equal to the waiting time. Note that the waiting time when making AR is smaller than when not making AR. However, it may be greater than zero, since it is possible that the server is already reserved at the desired service time. For simplification, we assume that the service period is long enough such that we can ignore the transient phase before the queue reaches its steady state.

In a preemptive-resume queue, if a job is interrupted, then it later resumes and is not restarted.  Due to this property, if the server is idle and a customer is waiting for service, the customer will be served even if service cannot be completed due to an existing reservation (in this case the service will be preempted and later resumed). Hence, supporting advance reservations in a preemptive-resume queue does not impact the utilization of the server which is $\rho=\lambda/\mu$. Figure \ref{fig_ex} illustrates the model.


\begin{figure}[t]
 \centering \includegraphics[scale=1.2]{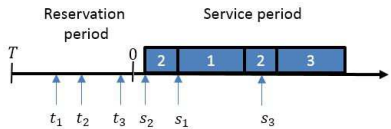}
  \caption{An illustration of the model with three customers.  Customer $1$ makes a reservation at time $t_1$, and is served upon arrival at $s_1$.  Customer $2$ also makes a reservation and is served upon arrival, but his/her service is preempted by customer $1$ which made a reservation earlier. Customer $3$ is served only when the service of customer $2$ is completed.}
    \label{fig_ex}
\end{figure}

Note that customers do not know a-priori what will be their waiting time if making or if not making AR. The decision is based on statistical information only, namely the values of $\lambda$, $\mu$ and $F_T$. However, once a customer decides to make a reservation, the system can provide him/her with the start and end times of the service.

\section{Equilibrium Analysis}\label{sec_analysis}
We can analyze this system as a priority queue where a priority between $0$ (lowest priority) and $1$ (highest priority)  is assigned to each customer. A customer with request time $t$ has priority $0$ if not making AR and priority $p=1-F_T(t)$ if making AR. Customers that share the same priority are served on a first-come-first-served basis. We refer to $p$ as the \textit{potential priority}. Due to the probability integral transformation theorem \cite[p. 320]{dodge2006oxford}, we know that $p$ is a random variable, uniformly distributed in $[0,1]$.

Since customers are statistically identical, we consider only symmetrical behavior. Thus, a decision of a tagged customer is a mapping of his/her potential priority $p$ to the probability of making AR. We denote this strategy function by $\sigma(p)$. Consider a tagged customer with potential priority $p$. We define $W(\cdot)$ to be a mapping of the strategy followed by the rest of the customers and the priority of the tagged customer to his/her expected waiting time. Thus, the expected waiting time of the tagged customer is $W(\sigma,p)$ if making AR and $W(\sigma,0)$ otherwise. Since customers are strategic, a customer with potential priority $p$ will make $AR$ only if
\begin{align}
W(\sigma,p)+C \le W(\sigma,0).
\end{align}

Next, we define a threshold strategy and show that this is the only strategy that can lead to equilibria.
\begin{definition}
Let $\tau\in(0,1]$. A strategy function $\sigma(p)$ is said to be a threshold strategy if it satisfies
\begin{equation}
\sigma(p)=\left\{\begin{array}{ll} 1 & \mbox{ if } p>\tau, \\ 0 & \mbox{ if } p\leq \tau. \\ \end{array}\right. \nonumber
\end{equation}
\end{definition}

\begin{lemma} \label{lemma_111}
At equilibrium, all customers follow a threshold strategy.
\end{lemma}

\begin{proof}
Consider any  strategy function $\sigma$. Since the expected waiting time  is non-increasing with the priority, there is either a single potential priority, or an interval of potential priorities, or no potential priority such that
\begin{align} \label{eq_170}
W(\sigma,p)+C=W(\sigma,0).
\end{align}
Note that the left hand side of \eqr{eq_170} is the expected total cost if making AR, while the right hand side is the expected total cost if not making AR. If Eq.~\eqref{eq_170} holds for a single value $p{'}$, then a customer with potential priority greater (respectively, smaller) than $p{'}$ is better off making (respectively, not making) AR. Therefore, $\sigma$ is an equilibrium strategy only if it is a threshold strategy with threshold $\tau=p{'}$.

If Eq.~\eqref{eq_170} holds for an interval of values $[p{'},p{''}]$, then all customers with potential priority $p\in [p{'},p{''}]$ do not make AR (otherwise, $W(\sigma,p)$ would not be a constant over that interval). Therefore, $\sigma$ is an equilibrium strategy only if it is a threshold strategy, with threshold $\tau=p{'}$.

Finally, suppose that Eq.~\eqref{eq_170} does not hold for any $p\in [0,1]$. If $W(\sigma,p)+C> W(\sigma,0)$ for all $p\in [0,1]$, then all customers are better off not making AR. Therefore, $\sigma$ is an equilibrium strategy only if it is a threshold strategy, with threshold $\tau=1$.

Note that a situation where $W(\sigma,p)+C< W(\sigma,0)$  for all $p\in [0,1]$ does not exist, since a customer with potential priority zero has the same expected waiting time if making or avoiding AR.
\end{proof}
Next, we define two types of equilibria.
\begin{definition}
An equilibrium strategy with threshold $\tau$ is called a \textit{some-make-AR} equilibrium if $\tau<1$.
\end{definition}

\begin{definition}
An equilibrium strategy with threshold $\tau$ is called a \textit{none-make-AR} equilibrium if $\tau=1$.
\end{definition}

Since the structure of the equilibrium depends on the reservation cost, we aim to determine the equilibrium to which a given reservation cost leads. Given that all customers follow a threshold strategy, we define a \textit{threshold customer} to be a customer with potential priority equals exactly to the threshold followed by all other customers.

Given a  strategy with threshold $\tau$, a threshold customer that makes AR observes three priority classes: 
\begin{enumerate}

\item A lower priority class which contains all customers with potential priority smaller than the threshold customer (none of them makes AR). The arrival rate of customers belonging to this class is $\lambda\tau$.
\item A priority class which contains only the threshold customer (since the potential priority is a continuous random variable, the probability that two customers will have the same potential priority is zero). Thus, the arrival rate of customers belonging to this class is $0$. 
\item A higher priority class which contains all customers with greater potential priority (they all made AR before the threshold customer). The arrival rate of customers belonging to this class is $\lambda(1-\tau)$.
\end{enumerate}
A threshold customer that does not make AR only observes two classes: 
\begin{enumerate}
\item A priority class which contains the threshold customer and all customers with smaller potential priority. The arrival rate of customers belonging to this class is $\lambda\tau$.
\item A higher priority class which contains all customers with greater potential priority. The arrival rate of customers belonging to this class is $\lambda(1-\tau)$.
\end{enumerate}

 Based on the priority classes defined above, we find the expected waiting of the threshold customer if making or not making AR. We apply the known formula of the waiting time in an $M/G/1$ queue with preemptive-resume priorities \cite[p.175]{conway2012theory} and obtain the following:
\begin{enumerate}
\item The expected waiting time of the threshold customer if making AR is
\begin{align}\label{eq-8}
W_{AR}(\tau)=\frac{\mu-\frac{\lambda}{2}\left(1-\tau\right)}{\left(\mu-\lambda\left(1-\tau\right)\right)^2}-\frac{1}{\mu}.
\end{align}
\item The expected waiting time of the threshold customer if not making AR is
\begin{align}\label{eq-9}
W_{AR'}(\tau)=\frac{\mu-\frac{\lambda}{2}}{\left(\mu-\lambda\left(1-\tau\right)\right)\left(\mu-\lambda\right)}-\frac{1}{\mu}.
\end{align}
\end{enumerate}

The condition for threshold $\tau<1$ to be a \textit{some-make-AR} equilibrium is
\begin{align}\label{eq-12}
C +W_{AR}(\tau) = W_{AR'}(\tau).
\end{align} That is, a customer with potential priority equals to the threshold is indifferent between the two actions.
The condition for threshold $\tau=1$ to be a \textit{none-make-AR} equilibrium is
\begin{align}\label{eq-120}
C + W_{AR}(1) \ge W_{AR'}(1).
\end{align} That is, a customer with potential priority $1$ (and hence, all customers) are better off not making AR.

By isolating $C$ in \eqr{eq-12}, we define $C(\tau)$ to be a function that maps a threshold to the reservation cost that leads to that threshold
\begin{align}\label{eq1}
C(\tau)\triangleq\frac{\lambda\cdot \mu\cdot \tau}{2\left(\mu-\lambda\right)\cdot \left(\mu-\lambda\left(1-\tau\right)\right)^2}.
\end{align}
We conclude that given reservation cost $C$, the  threshold $\tau^e\in (0,1)$ represents a \textit{some-make-AR} equilibrium if and only if $C=C(\tau^e)$. The threshold $\tau^e=1$ represents a \textit{none-make-AR} equilibrium if and only if $C\ge C(1)$. In order to find the equilibrium structure, we next find the properties of $ C(\tau)$. 
\begin{lemma}
If $\rho\le1/2 $, then $C(\tau)$ is a monotonically increasing function. If $\rho>1/2 $, then $C(\tau)$ is a unimodal function with a global maximum.
\end{lemma} 

\begin{proof}
First, we compute the derivative of $C(\tau)$:
\begin{align}
\frac{d C}{d \tau}=\frac{\lambda \mu\left(\lambda\left(1+\tau\right)-\mu\right)}{2\left(\lambda-\mu\right)\left(\mu-\lambda\left(1-\tau\right)\right)^3}.
\end{align}
Since the denominator is negative for any $\tau$, the sign of the derivative is determined by the sign of $\lambda \left(1+\tau \right)-\mu$. If $\rho\le1/2$, then this expression is negative  for any $\tau\in (0,1)$ and the derivative  of $C(\tau)$ is positive for any $\tau\in (0,1)$. If $\rho>1/2$, then the derivative of $C(\tau)$ is positive for any  $\tau<(\mu-\lambda) / \lambda$; is  equal to zero at $\tau=(\mu-\lambda) / \lambda$; and negative otherwise. Thus, for any value of $\rho>1/2$, $C(\tau)$ is unimodal with a global maximum. 
\end{proof}

Next, we define:
\begin{align}\label{eq-1}
\underline{C}\triangleq C(1)=\frac{\lambda}{2\mu \left(\mu-\lambda\right)},
\end{align}
and 
\begin{align}
\overline{C}\triangleq\frac{\mu}{8\left(\lambda-\mu\right)^2}.
\end{align}
Note that if $\rho\le 0.5$, then $\underline{C}$ is the maximum value of $C(\tau)$ and if $\rho>0.5$, then $\overline{C}$ is the maximum value of $C(\tau)$. We can now state the main result of this section:
\begin{theorem}\label{theorem1}
The game has the following equilibrium structure.\\
When $\rho\le 1/2$:
\begin{itemize}
\item If $C<\underline{C}$, then there is a unique \textit{some-make-AR} equilibrium.
\item If $C>\underline{C}$, then there is a unique \textit{none-make-AR} equilibrium.
\end{itemize}
When $\rho>1/2$:
\begin{itemize}
\item If $C<\underline{C}$, then there is a unique \textit{some-make-AR} equilibrium.
\item If $\underline{C}<C<\overline{C}$, then there are two \textit{some-make-AR} equilibria and a \textit{none-make-AR} equilibrium.
\item If $C>\overline{C}$, then there is a unique \textit{none-make-AR} equilibrium.
\end{itemize}
\end{theorem}

\begin{proof}
We begin with $\rho\le 0.5$. If $C<\underline{C}$, then there is a single value of $\tau$ such that $C=C(\tau)$ has a solution. Hence, there is one \textit{some-make-AR} equilibrium. A \textit{none-make-AR} equilibrium does not exist since $C(1)>C$. If $C>\underline{C}$, then there is no value of $\tau$ such that $C=C(\tau)$ has a solution. Hence, a \textit{some-make-AR} equilibrium does not exist. On the other hand, a \textit{none-make-AR} equilibrium exists since $C>C(1)$.

Next, consider $\rho>0.5$. In the range $[0,\underline{C}]$, the function $C(\tau)$ is monotonically increasing. Thus, if $C\in [0,\underline{C}]$, then there is a single value of $\tau$ such that $C=C(\tau)$ has a solution, and hence there is one \textit{some-make-AR} equilibrium. In the range $[\underline{C},\overline{C}]$, the function $C(\tau)$ is unimodal. Thus, if $C\in [\underline{C},\overline{C}]$, then there exist two values of $\tau$ that solve $C=C(\tau)$, and hence there are two \textit{some-make-AR} equilibria. The condition for the existence of  a \textit{none-make-AR} equilibrium is the same as in the case of $\rho\le 0.5$.
\end{proof}

\section{Revenue Maximization}\label{sec:rev}
In this section, we assume that the reservation cost is a fee determined by the service provider. We show that the fee that maximizes the revenue leads to a unique equilibrium if the utilization is smaller than $2/3$ and to multiple equilibria if the utilization is greater than $2/3$.  We also show that the revenue from AR fee depends only on the utilization of the queue and the fee itself. Thus, if the demand and the number of servers increase proportionally, then the revenue from AR fees does not change.

The revenue per time unit, at equilibrium with threshold $\tau^e$, is the number of customers making AR multiplied by the AR fee that leads to that equilibrium. The expected revenue is
\begin{align}\label{eq__1}
R(\tau^e)\: & =\lambda(1-\tau^e)C(\tau^e)\nonumber \\ & = \frac{\lambda^2(1-\tau^e)\tau^e\mu}{2(\mu-\lambda)(\lambda(\tau^e-1)+\mu)^2}.
\end{align} 

With some manipulation, we get that the revenue function does not depend on the values of $\lambda$ and $\mu$ but only on the utilization $\rho$:
\begin{align}\label{eq2}
R(\tau^e)=\frac{\rho(1-\tau^e)\tau^e}{2(1-\rho)(1+\rho(\tau^e-1))^2}.
\end{align}
At first glance, this result seems surprising since it implies that the revenue does not increase when scaling the system (i.e., increasing both arrival and service rates). However, in an $M/D/1$ queue, the waiting time decreases  as the system gets larger, and hence customers are less motivated to make AR. Therefore, scaling the system has a trade off.  For a given threshold, as we scale the system, more customers will make AR but they will  pay a smaller fee.

By solving the equation $d R / d \tau^e =0$, we find that the optimal threshold is ${\tau^{opt}=(1-\rho )/(2-\rho)}$. By substituting $\tau^{opt}$ into \eqr{eq1}, we get that the optimal fee is
\begin{align} \label{eq_2}
C^*=\frac{\lambda( 2\mu-\lambda)}{8\mu(\mu-\lambda)^2}.
\end{align}
Similarly, by substituting $\tau^{opt}$ into \eqr{eq2}, we get that the maximum possible revenue is
\begin{align}
R^*=\frac{\rho^2}{8(1-\rho)^2}.
\end{align}

Next we find the number of equilibria  when $C=C^*$.
\begin{theorem}\label{corMD1}
The revenue maximizing fee $C^*$ leads to a unique \textit{some-make-AR} equilibrium if $\rho<2/3$ and to multiple equilibria, including a \textit{none-make-AR} equilibrium, otherwise.
\end{theorem}

\begin{proof}
	The optimal reservation cost $C^*$ leads to multiple equilibria only if $\rho > 0.5$ and $C^*>\underline C$ (see Theorem \ref{theorem1}). Using  \eqr{eq-1} and \eqr{eq_2}, we deduce that if $C^*>\underline C$ , then
\begin{align}
\frac{\lambda( 2\mu-\lambda)}{8\mu(\mu-\lambda)^2}>\frac{\lambda}{2\mu \left(\mu-\lambda\right)}.
\end{align}
One can show that the inequality above holds only if $\rho>2/3$. 
\end{proof}
Figure \ref{fig_rev} illustrates the game outcome when $C=C^*$.

\begin{figure} [t]
\centering
  \begin{subfigure}[b]{1.2\textwidth}
\centering
    \includegraphics[scale=0.4]{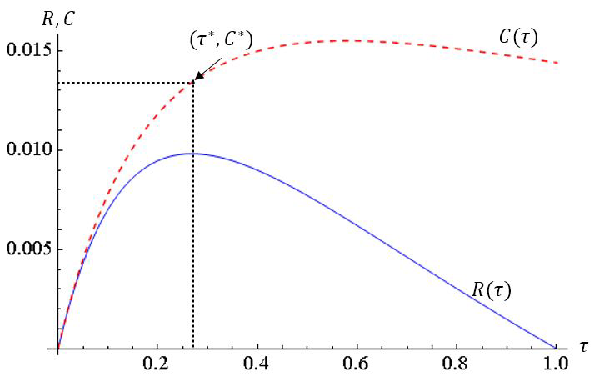}
    \subcaption{$\lambda=38$, $\mu=60$}
    \label{fig1}
  \end{subfigure}
  \begin{subfigure}[b]{0.47\textwidth}
 \centering
    \includegraphics[scale=0.4]{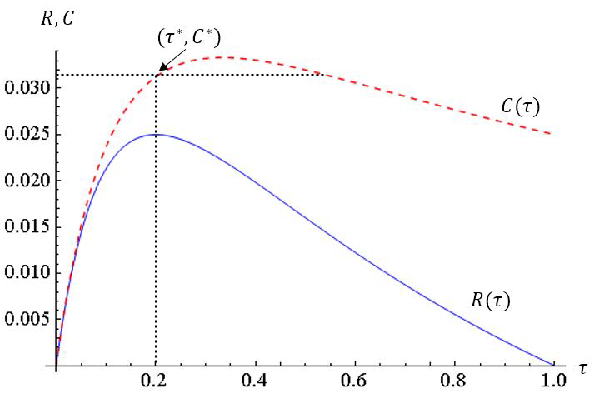}
    \subcaption{$\lambda=45$, $\mu=60$}
    \label{fig2}
  \end{subfigure}\\
  \caption{When the utilization $\rho<2/3$, the optimal fee $C^*$ leads to a unique equilibrium (a). When $\rho>2/3$, $C^*$ leads to multiple equilibria (b).}
  \label{fig_rev}
  \end{figure}

\subsection{Price of Conservatism}\label{poc}
Assuming that $\rho>2/3$, the provider can either be risk-averse and charge a fee that leads to a unique equilibrium with guaranteed revenue, or it can be risk-taking and charge a higher fee that may lead to greater revenue but also to zero revenue.
To compare between the two options, we use the Price of Conservatism (PoC) metric, which was introduced in \cite{simhon2017advance}. PoC is the ratio between the maximum possible revenue $R^*$ and the maximum guaranteed revenue $R_g^*$, which is defined as follows.
\begin{align}
R_g^*=  \quad &\; \sup_{0<\tau^e <1} R(\tau^e).
\nonumber \\ & \mbox{s.t. } C(\tau^e) < \underline{C}.
\end{align}

Since $R(\tau^e)$ has exactly one extreme point (which is $\tau^{opt}$), it is increasing in the range $[0,\tau^{opt})$. Therefore,  the maximum guaranteed revenue is achieved when choosing the largest $\tau^e$ for which   $C(\tau^e) < \underline{C}$. In other words,  $C$ should  be slightly smaller than $\underline{C}$. By solving $C(\tau)=\underline{C}$, we get two solutions: $\tau^{e1}=1$ and 
\begin{align}\label{eq__2}
\tau^e_2=\left(\frac{1-\rho}{\rho}\right)^2.
\end{align}
%
By substituting $\tau^e_2$ into \eqr{eq2}, we get
\begin{align}
R_g^*=\frac{2\rho-1}{2(1-\rho)},
\end{align}
and by dividing  $R^*$ by $R_g^*$, we get
\begin{align}
PoC=\frac{\rho^2}{-8\rho^2+12\rho-4}.
\end{align}
We conclude with the following theorem.
\begin{theorem}
If $\rho<2/3$, then $PoC=1$. Else, $PoC=\frac{\rho^2}{-8\rho^2+12\rho-4}$. 
\end{theorem}
Figure \ref{fig_ex_22} shows the maximum possible revenue and the maximum guaranteed revenue in a system with parameters $\lambda=45$ and $\mu=60$. 
\begin{figure}[t]
  \centering \includegraphics[scale=1.2]{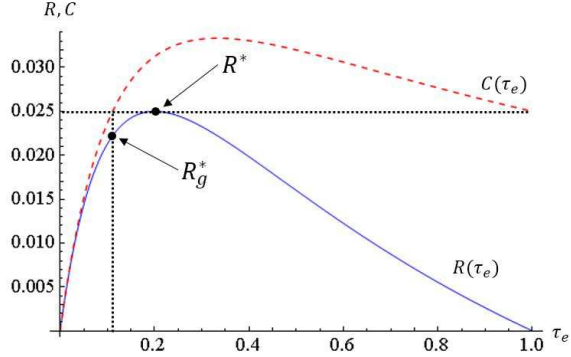}
  \caption{The maximum possible revenue and the maximum guaranteed revenue in a system with parameters $\lambda=45$ and $\mu=60$.}
    \label{fig_ex_22}
\end{figure}

By computing the derivative of PoC with respect to $\rho$, we get that for any $\rho>2/3$
\begin{align}
\frac{d PoC}{d \rho}=\frac{\rho(3\rho-2)}{4(2\rho^2-3\rho+1)^2}>0
\end{align}
Thus, we obtain the following corollary:
\begin{corollary}
The price of conservatism increases with the utilization.
\end{corollary}
\noindent That is, as the utilization increases, the ratio between the potential revenue when the provider is risk-taking and the revenue when the provider is risk-averse increases and tends to $\infty$ as $\rho \to 1$.

\section{Dynamic Games}\label{sec_learning}

\subsection{ Learning Models}
In this section, we study dynamic versions of the game. In dynamic games (also known as learning models, since players learn over time the behavior of other players), it is assumed that the game repeats many times and that initially players do not necessarily follow an equilibrium strategy. The goal is to find the long-term behavior of the customers. In our analysis, we use a \textit{best response dynamic} model which is rooted in Cournot study of duopoly \cite{cournot1897recherches}.  we next describe the learning models.

At each step (game), a new set of customers participate (or the same set of participants but with new realizations of request times). At the first step, all customers have an \textit{initial belief} about the strategy that is followed by all customers. Next, we assume:
\begin{assumption} \label{ass_1}
Customers that are indifferent between actions $AR$ and $AR'$ choose action AR'.
\end{assumption}
\noindent Based on this assumption, and using the proof of Lemma \ref{lemma_111}, one can show that the best response of all customers to any initial belief  is a threshold strategy. In order to simplify the analysis and since a threshold strategy is followed at all steps, we also assume:
\begin{assumption}
The initial belief is a threshold strategy. 
\end{assumption}

We denote by $\tau_i\in[0,1]$ the threshold of the strategy followed at step $i\ge 1$. We denote by $\hat{\tau}_i$ the estimation of this strategy and we distinguish between two types of learning: 
\begin{enumerate}
\item \textbf{Strategy learning}. In this type of learning, the analysis assumes that at each step $i$, $\hat{\tau}_i \triangleq\tau_i$. That is, customers observe past strategies.
\item \textbf{Action learning}. In this type of learning, customers observe previous actions and use the proportion of customers that chose $AR$ at the previous step as an estimation of the strategy that was followed at that step. Namely, if the demand and the number of reservations at step $i$ are $d_i$ and  $d^{AR}_i$ respectively, then
\begin{align}\label{eq_700}
\hat{\tau}_i\triangleq1-\frac{d^{AR}_i}{d_i}.
\end{align}
\end{enumerate}

Since the best response of all customers to any belief is a threshold strategy, we can define a joint best response function $BR: [0,1] \rightarrow [0,1]$. The input is a belief about the threshold strategy that will be followed by all customers. The output is the best response threshold to that belief. Thus, we can describe the best response dynamics of the game as the following process:
\begin{align} \label{eq_40}
\; & \hat{\tau}_1=\beta, 
 \\ & \tau_i=BR\left( \hat{\tau}_{i-1}\right),\quad \forall i>1,
\end{align}
where $\beta\in[0,1]$ represents the initial belief. Note that under \textit{strategy-learning} this process is deterministic, while under \textit{action-learning} this process is a Markov process \cite[Chapter 3]{gardiner1985handbook}. In the following sections we analyze this dynamic process.

Next, we focus on the behavior of customers at a given step. Thus, we remove the subscript $i$. We begin the analysis with the following observations: 
\begin{enumerate}[(i)]
\item Given a belief $\beta$ (i.e., assuming that all other customers follow the threshold $\beta$), if a tagged customer with potential priority $p>\beta$ chooses $AR$, then all customers with greater potential priority have higher priority and all customers with smaller potential priority have lower priority. Therefore, the (believed) expected waiting time of the tagged customer is equal to the expected waiting time of a threshold customer that chooses $AR$ in a system where all customers follow the threshold $p$. Hence, 
\begin{align}\label{eq_121}
W(\beta,p) =W_{AR}(p), \quad \mbox{if } p \ge \beta,
\end{align}
where $W_{AR}(\cdot)$ is defined in \eqr{eq-8}.
\item Given a belief $\beta$, if a tagged customer with potential priority $p<\beta$ chooses $AR$, then his/her (believed) expected waiting time is the same as the expected waiting time of the threshold customer (recall that each customer believes that he/she is the only one deviating). Hence,
\begin{align}\label{eq_122}
W(\beta,p) =W_{AR}(\beta), \quad \mbox{if } p < \beta.
\end{align}
\item The expected waiting time of all customers that choose $AR'$ are equal. Hence, 
 \begin{equation} \label{eq_12}
 W(\beta,0)= W_{AR'}(\beta),\quad \forall p\in[0,1],
 \end{equation}
 where $W_{AR'}(\cdot)$ is defined in \eqr{eq-9}.
\end{enumerate}
\noindent Those properties will be used later to prove our main results. Next, we separately explore the case of a unique \textit{some-make-AR} equilibriun and the case of multiple equilibria.
 
\subsection{Learning with Unique Some-make-AR Equilibrium}
Consider a \textit{some-make-AR} equilibrium with equilibrium threshold $\tau^e$. By computing the derivative of $W_{AR}(\beta)$ and $W_{AR'}(\beta)$, one can verify that both functions are decreasing with $\beta$. This property will be used in the proof of the following lemma.
\begin{lemma} \label{lemma_2}
Under a unique some-make-AR  equilibrium:
\begin{enumerate}
\item  If a belief $\beta\in [0,\tau^e)$ , then  $BR(\beta) \in (\beta,\tau^e)$. 
\item  If a belief $\beta\in (\tau^e,1]$, then  $BR(\beta)=0$.
\end{enumerate}
\end{lemma}

\begin{proof}
From \eqr{eq-8} and \eqr{eq-9}, we deduce that $W_{AR}(0)+C>W_{AR'}(0)$. Since, under unique \textit{some-make-AR} equilibrium, $W_{AR}(0)+C$ and $W_{AR'}(0)$ intersect once, we conclude that
\begin{align} \label{eq_500}
W_{AR}(\beta)+C>W_{AR'}(\beta) \quad \forall \beta\in [0,\tau^e)
\end{align}
and
\begin{align} \label{eq_505}
W_{AR}(\beta)+C<W_{AR'}(\beta) \quad \forall \beta\in (\tau^e,1).
\end{align}
For proving part 1, assume that $\beta<\tau^e$.  Since $W_{AR}(\tau^e)+C=W_{AR'}(\tau^e)$, and since $W_{AR}(\cdot)$ is a decreasing function we deduce that
\begin{align} \label{eq_501}
W_{AR}(\tau^e)+C<W_{AR'}(\beta).
\end{align}

From \eqr{eq_500} and \eqr{eq_501} we deduce that there exists a $\tau\in(\beta,\tau^e)$ such that 
\begin{align} \label{eq_502}
W_{AR}(\tau)+C=W_{AR'}(\beta).
\end{align}

Given the equation above  and using \eqr{eq_121} and \eqr{eq_122}, one can see that all customers with potential priority $p>\tau$ choose $AR$, while all customers with potential priority $p\le \tau$ choose $AR'$. This complete the proof of the first part of the lemma.

Now, let assume that $\beta >\tau^e$. Based on \eqr{eq_121} and \eqr{eq_122} and since $W_{AR}(\cdot)$ is a decreasing function, we deduce that
\begin{align}\label{eq-66}
W(\beta, p)\le W_{AR}(\beta),\quad \forall p \in [0,1].
\end{align}
From Eq.~\eqref{eq_505} and \eqr{eq-66}, we deduce that
\begin{align}
W(\beta, p)+C<W_{AR'}(\beta),\quad \forall p \in [0,1].
\end{align}
That is, all customers are better off choosing $AR$ and the best response to $\beta$ is $\tau=0$.  
\end{proof}

Next, we study the long-term outcome of the dynamic game and establish the following result.
\begin{theorem}\label{theorem_5}
A game with unique some-make-AR equilibrium converges to equilibrium under \textit{strategy-learning} and cycles under \textit{action-learning}.
\end{theorem}

\begin{proof}
We begin with the first part of the theorem. Let assume that the initial belief $\beta\in [0,\tau^e)$. From Lemma \ref{lemma_2}, we deduce that $\tau_1 \in(\beta,\tau^e)$. Hence, $\tau_2 \in (\tau_1,\tau^e)$. By induction, we deduce that, for any $i\ge 0$,
\begin{flalign}
&\; \tau_i \ge \tau_{i-1}, \\ &
\tau_i \le \tau^e.
\end{flalign}
The set $\lbrace \tau_i, i=1,2...\rbrace$ is a monotonically increasing sequence bounded by $\tau^e$. Thus, it has a limit, denoted by $L$. since $\lim_{i\to \infty} \tau_i \rightarrow L$, and sense $\lim_{i\to \infty} BR(\tau_i)=L$, we conclude that the limit $L$ is a fixed point of $BR$, and hence it must be the equilibrium point $\tau^e$.

Next, we assume that $\beta \in (\tau^e,1]$. In this case, based on Lemma \ref{lemma_3}, $\tau_1=0$ and the game converges to equilibrium as in the case of $\beta\in [0,\tau^e)$.

From \eqr{eq_700}, we deduce that, under \textit{action-learning}, at any step $i$, if $\tau_i>0$, then $P(\hat \tau_i>\tau^e)>0$ (i.e., if the strategy followed at step $i$ is greater than zero then there is a positive probability that the fraction of customers not making AR will be greater than $\tau^e$). Once $\hat \tau_i>\tau^e$, then $\tau_{i+1}=0$. Thus, customers strategy cycles between $0$ and $\tau^e$.  
\end{proof}

Next, we determine, for a given reservation cost, whether a service provider who wishes to maximize the number of reservations is better off under \textit{strategy-learning} or under \textit{action-learning}.

\begin{theorem}\label{theorem_33}
In a dynamic game with a unique \textit{some-make-AR} equilibrium, the average number of customers making AR under \textit{action-learning} is greater than under \textit{strategy-learning}.
\end{theorem}
\begin{proof}
Denote the unique equilibrium by $\tau^e$. Consider an arbitrary step $i$ and assume that \textit{action-learning} is applied. If at step $i$ the strategy $\hat \tau_i \in [\tau^e,1]$, then all customers will choose $AR$ at the next step. If $\tau_i \in (\tau_i, \tau^e)$, then $\tau_{i+1}\in (0, \tau^e)$. Thus, in any realization, the strategy followed by all customers in all steps is a random variable that takes values between $0$ and $\tau^e$. In \textit{strategy-learning}, the strategy followed by all customers converges to $\tau^e$. Thus, the average fraction of customers not making $AR$ converges to a value between $0$ and $\tau^e$ under \textit{action-learning} and to $\tau^e$ under \textit{strategy-learning}. 
\end{proof}
Next, we present a simulated example that compares between the revenue under \textit{action-learning} and under \textit{strategy-learning}. The pseudo-code of the simulation is given in Algorithm 1. The inputs of the procedure are the arrival rate $\lambda$, the initial belief $\beta$, the reservation cost $C$ and the number of steps $l$.
\begin{algorithm}[H] 
\caption{Learning Simulation ($\lambda,\beta,C,l$)}
\begin{algorithmic}
\STATE{$\hat{\tau_1}=\beta$}
\FOR{$i \leftarrow 1$ to $l$ \COMMENT{iterating over all steps}} 
\STATE{ $D_{AR} \leftarrow 0$   \COMMENT{variable counting the number of reservations}\\
$D \leftarrow$ \textit{generate Poisson random variable} \COMMENT{the number of customers}
\FOR{$j \leftarrow 1$ to $D$ \COMMENT{iterating over all customers}}
\STATE{
$p\leftarrow$ \textit{generate random variable from U(0,1)} \COMMENT{the potential priority}\\
\IF{$p>\hat{\tau_i}$ \COMMENT{check if the potential priority is greater than the current belief}}
\IF{$W_{AR}(p)+C<W_{AR'}(\hat{\tau_i})$ \COMMENT{check if the customer is better off making AR}}
\STATE{$D_{AR}\leftarrow D_{AR} +1$} \COMMENT{increase the number of reservations by one}
\ENDIF
\ELSE{
\IF{$W_{AR}(\hat{\tau_i})+C<W_{AR'}(\hat{\tau_i})$  \COMMENT{check if the customer is better off making AR}}
\STATE{$D_{AR}\leftarrow D_{AR} +1$} \COMMENT{increase the number of reservations by one}
\ENDIF}
\ENDIF
}
\ENDFOR \\
\IF{strategy-learning}\vspace{-13pt} 
\STATE{\begin{flalign*}
\hat{\tau}_{i+1}\leftarrow \left\{\begin{array}{ll} 0 & \mbox{ if } \hat{\tau_i}>\tau^e, \\ \tau: W_{AR}(\tau)+C=W_{AR'}(\hat{\tau_i})  & \mbox{ if } \hat{\tau_i}\leq \tau^e, \\ \end{array}\right. \nonumber  &&\end{flalign*}} \COMMENT{compute the current strategy\vspace{6pt}}
\ENDIF
\IF{action-learning}\vspace{6pt}
\STATE{$\hat{\tau}_{i+1}\leftarrow 1- \frac{D_{AR}}{D}$} \COMMENT{estimate the current strategy} \vspace{6pt} 
\ENDIF
}
\ENDFOR
\end{algorithmic}
\end{algorithm}

\begin{example} \label{ex_1}
Consider a queue with parameters $\lambda=45$ and $\mu=60$. Let the reservation cost be $C=0.024$. The unique equilibrium (computed using \eqr{eq-12} and \eqr{eq-120}) is $\tau^e=0.1026$ (i.e., on a static game, on average, $89.74\%$ of the customers make AR). We run a simulation of $10,000$ steps. Each step lasts for one time unit (i.e., the average demand at each step is $45$).  We set the initial belief to be $\beta=\tau^e$. The average number of reservations per time unit is $40.3$ (i.e., on average, $89.5\%$ make $AR$) under \textit{strategy-learning} and $42.7$  (i.e., on average, $94.9\%$ make $AR$) under \textit{action-learning}. We conclude that, as Theorem \ref{theorem_33} states, when customers base their decisions on historic actions and not strategies, more customers make AR. Statistical analysis (one-tailed t-test) shows that the difference between the mean number of reservations under \textit{action-learning} and under \textit{strategy-learning} is statistically significant, with confidence level of $99\%$. In Figure~\ref{fig_ex1}, we plot the number of reservations, under \textit{action-learning} and under \textit{strategy-learning}. We use the same realization of customer arrivals in each case and we can see that at each iteration, the number of reservations is greater (or equal) under \textit{action-learning}.
\end{example}
\begin{figure}[t] 
 \centering \includegraphics[scale=1.2]{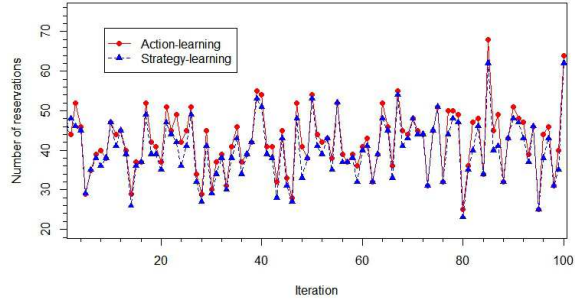}
  \caption{Simulation results. In a game with unique equilibrium, more customers make AR under \textit{action-learning} than under \textit{strategy-learning}.}
 \label{fig_ex1}
\end{figure}
We conclude that if the provider interest is that as many customers as possible will make reservations, then it is better off if customers gain information about previous actions rather than strategies.

\subsection{Learning with Multiple Equilibria}

\begin{lemma} \label{lemma_3}
In a game with multiple equilibria with thresholds $\tau^{e1}$, $\tau^{e2} $ and $1$:
\begin{enumerate}
\item  If a belief $\beta\in [0,\tau^{e1})$ , then  $BR(\beta) \in (\beta,\tau^{e1})$. 
\item  If a belief $\beta\in (\tau^{e1},\tau^{e2})$, then  $BR(\beta)=0$.
\item  If a belief $\beta\in (\tau^{e2},1)$ , then  $BR(\beta) \in (\beta,1]$. 
\end{enumerate}
\end{lemma}

\begin{proof}
From \eqr{eq-8} and \eqr{eq-9}, we deduce that $W_{AR}(0)+C>W_{AR'}(0)$. Since, under unique \textit{some-make-AR} equilibrium, $W_{AR}(0)+C$ and $W_{AR'}(0)$ intersect twice at $\tau^{e1}$ and $\tau^{e2} $, we conclude that
\begin{align} \label{eq_520}
W_{AR}(\beta)+C>W_{AR'}(\beta) \quad \forall \beta\in \{[0,\tau^{e1}),(\tau^2_e,1)\}
\end{align}
and
\begin{align} \label{eq_525}
W_{AR}(\beta)+C<W_{AR'}(\beta) \quad \forall \beta\in (\tau^{e1},\tau^{e2}).
\end{align}

For proving part 1, assume that $\beta<\tau^{e1}$.  Since $W_{AR}(\tau^{e1})+C=W_{AR'}(\tau^{e1})$, and since $W_{AR}(\cdot)$ is a decreasing function we deduce that
\begin{align} \label{eq_511}
W_{AR}(\tau^{e1})+C<W_{AR'}(\beta).
\end{align}

From \eqr{eq_520} and \eqr{eq_511} we deduce that there exists a $\tau\in(\beta,\tau^{e1})$ such that 
\begin{align} \label{eq_512}
W_{AR}(\tau)+C=W_{AR'}(\beta).
\end{align}

Given the equation above  and using \eqr{eq_121} and \eqr{eq_122}, one can see that all customers with potential priority $p>\tau$ choose $AR$, while all customers with potential priority $p\le \tau$ choose $AR'$. This complete the proof of the first part of the lemma.

Now, assume that $\beta \in(\tau^{e1},\tau^{e2})$. Based on \eqr{eq_121} and \eqr{eq_122} and since $W_{AR}(\cdot)$ is a decreasing function, we deduce that
\begin{align}\label{eq-66}
W(\beta, p)\le W_{AR}(\beta),\quad \forall p \in [0,1].
\end{align}
From Eq.~\eqref{eq_525} and \eqr{eq-66}, we deduce that
\begin{align}
W(\beta, p)+C<W_{AR'}(\beta),\quad \forall p \in [0,1].
\end{align}
That is, all customers are better off choosing $AR$ and the best response to $\beta$ is $\tau=0$.
The third part of the lemma can be proved using the same arguments as in the proof of the first part of the lemma. 
\end{proof}

Next, we study the long-term outcome of a dynamic game with multiple equilibria.
\begin{theorem}\label{theorem_5}
A game with multiple equilibria converges to some-make-AR or none-make-AR equilibrium (depend on the initial belief) under \textit{strategy-learning} and to none-make-AR equilibrium under \textit{action-learning}.
\end{theorem}

\begin{proof}
Using the same arguments as in the proof of Theorem \ref{theorem_5} and based on Lemma \ref{lemma_3}, one can show the following. Under \textit{strategy-learning}, a game with initial belief $\beta<\tau^{e2}$ converges to $\tau^{e1}$, while a game with initial belief $\beta\in(\tau^{e2},1]$ converges to $1$.

Under \textit{action-learning}, if at some step $i$, $\tau_i=1$ and \textit{none-make-AR} is an equilibrium, then at all future steps all customers will keep not making AR. Given any threshold strategy $\tau>0$ followed by all customers, there is a positive probability that the potential priority of all customers will be smaller than $\tau$, and hence none of the customers will make AR. If the game repeats infinite many times, then with probability one, at some point, none of the customers will make AR and the game will converge to \textit{none-make-AR} equilibrium.  
\end{proof}

\begin{example} \label{ex_2}
Consider a queue with parameters $\lambda=45$ and $\mu=60$. Let the reservation cost be $C=0.032$. Using \eqr{eq-12} and \eqr{eq-120} we compute the set of equilibria:  $ \tau^{e1}=0.22,\tau^e_2=0.5 ,\tau^e_3=1$. We set three different initial strategies: $\beta_1=0.2,\beta_2=0.4$ and $\beta_3=0.6$. We apply \textit{strategy-learning}. As Figure \ref{fig_ex_2} shows, within a few steps, the system converges to an equilibrium.
\end{example}

\begin{figure}[t]
  \centering \includegraphics[scale=1.2]{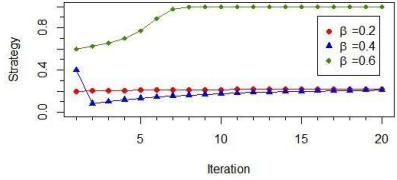}
  \caption{Convergence to equilibrium under \textit{strategy-learning}}
    \label{fig_ex_2}
\end{figure}
%

\subsection{Profit Maximization in Dynamic Games}
In this section, we assume that the reservation cost is a fee collected by the provider. Our goal is to find the fee that maximizes the provider revenue in the dynamic games setting. Under action learning, any fee that leads to multiple equilibria will eventually lead to zero revenue. Hence, we focus, in this section, on \textit{strategy-learning}.

Under \textit{strategy-learning} with multiple equilibria, the initial belief determines to which equilibrium the game will converge. To execute the analysis, we assume that the initial belief $\beta$ is a continues random variable that takes values between zero and one.

Consider a game with multiple equilibria $\{\tau^{e1}, \tau^{e2}, 1\}$. From Lemma \ref{lemma_2}, Lemma \ref{lemma_3} and Theorem \ref{theorem_5}, we deduce that if the initial belief is in $[0,\tau^{e2})$, then the game converges to $\tau^{e1}$, otherwise it converges to $1$ (i.e., zero reservations). Thus, with probability $\mathbb{P}(\beta<\tau^{e2})$ the strategy converges to $\tau^{e1}$ and with probability $\mathbb{P}(\beta>\tau^{e2})$ it converges to $1$. Thus, the excepted revenue of the dynamic game at steady state is 
\begin{align}\label{eq-667}
R_D(\tau^{e1},\tau^{e2})=\mathbb{P}(\beta<\tau^{e2})R(\tau^{e1}).
\end{align}
where $R(\cdot)$ is defined in \eqr{eq__1}. Since the expected revenue depends on both $\tau^{e1}$ and $\tau^{e2}$, we next find the relation between those two thresholds. By manipulating the equation $C(\tau^{e1})=C(\tau^{e2})$ (see \eqr{eq1} for definition of $C(\cdot)$), we get the following relation.
\begin{align}\label{eq-3}
\tau^{e2}=\left(\frac{1-\rho}{\rho} \right)^2\frac{1}{\tau^{e1}}.
\end{align}


Given the distribution of $\beta$ and using \eqr{eq-667}, \eqr{eq-3} and \eqr{eq2}, one can find the value of $\tau^{e1}$ that maximizes the revenue and, in turn, the optimal fee. For instance, let assume that $\beta$ is uniformly distributed in $[0,1]$. In this case, the revenue as a function of $\tau^{e1}$ is
\begin{align}
R_D(\tau^{e1})=\frac{(1-\tau^{e1})(1-\rho)}{2 \rho(1-\rho(\tau^{e1}-1))^2}.
\end{align} 

By computing the derivative of $R_D(\tau^{e1})$ with respect to $\tau^{e1}$, one can show that it decreases with  $\tau^{e1}$. Thus, when considering multiple equilibria, the optimal value of $\tau^{e1}$ is  ${\{ \min \tau^{e1} |\underline{C}<C(\tau^{e1})<\overline{C} \}}$. From \eqr{eq__2} we know that this value is $((1-\rho)/\rho)^2$ and is obtained when $C=\underline{C}$

Combining this result with Corollary \ref{corMD1} leads to the following theorem:
\begin{theorem}\label{theorem_333}
Under \textit{strategy-learning}, if the initial belief is uniformly distributed between $0$ and $1$, then the optimal fee is $\underline{C}$ when $\rho>2/3$ and $C^*$ when $\rho<2/3$.
\end{theorem}

\section{Conclusion and future work} \label{conclusions}
In this paper, we analyzed an M/D/1 queue that supports advance reservations. We associated the act of making reservation with a fixed reservation cost  and studied the impact of this cost on the behavior of customers. First, we showed that if the utilization of the queue is greater than $1/2$, then there is a range of reservation costs that lead to multiple equilibria including one where no customer makes a reservation. Furthermore, if the utilization is greater than $2/3$ and the reservation cost is a fee charged by the service provider, then the fee  value that maximizes the revenue from AR belongs to the aforementioned range. In order to evaluate whether the provider should charge a lower fee with guaranteed revenue or a higher but riskier fee (yielding several equlibria) we used the price of conservatism (PoC) metric and found the ratio between the two revenues. Specifically, when the utilization exceeds $2/3$, we showed that the PoC increases with the utilization and tends to infinity as the utilization approaches 1.

In the second part of the paper, we studied a dynamic version of the game. We showed that if the customers observe  previous strategies, then the game converges to an equilibrium. If the customers observe previous actions, then the game converge to a none-make-AR equilibrium, if such an equilibrium exists, and cycles otherwise. Finally, we develop a method to derive the revenue-maximizing fee under dynamic games. This method helps to determine the optimal control parameters  in a game with many equilibria. We expect the same kind of methods to prove useful for the analysis of other types of dynamic games with many equilibria.


\bibliographystyle{ormsv080}
\bibliography{learning_bib1}

\end{document}